\documentclass[conference]{IEEEtran}
\IEEEoverridecommandlockouts

\usepackage{cite}
\usepackage{amsmath,amssymb,amsfonts,amsthm}
\usepackage{mathabx}
\usepackage{stmaryrd}
\usepackage{multirow}
\usepackage{float}
\usepackage{hyperref}
\usepackage{graphicx}
\usepackage{textcomp}
\usepackage{algorithm, algpseudocode}
\usepackage{xcolor,tikz, circuitikz}
\def\BibTeX{{\rm B\kern-.05em{\sc i\kern-.025em b}\kern-.08em
    T\kern-.1667em\lower.7ex\hbox{E}\kern-.125emX}}

\newcommand{\ket}[1]{|{#1}\rangle}

\newcommand{\ketbra}[1]{|{#1}\rangle\langle{#1}|}
\newcommand{\idg}[1]{{\bfseries #1)}}
\newcommand{\subfigimg}[3][,]{%
	\setbox1=\hbox{\includegraphics[#1]{#3}}% Store image in box
	\leavevmode\rlap{\usebox1}% Print image
	\rlap{\hspace*{2pt}\raisebox{\dimexpr\ht1-0.5\baselineskip}{{\bfseries \large\textsf{#2}}}}% Print label
	\phantom{\usebox1}% Insert appropriate spcing
}

\DeclareMathOperator{\tr}{tr}

\newtheorem{theorem}{Theorem}
\newtheorem{corollary}{Corollary}

\begin{document}

\title{Faster Optimal Decoder for Graph Codes with a Single Logical Qubit}
% Faster Optimal Decoder for LDPC Hypergraph Codes for Multiple Logical Qubits

\author{\IEEEauthorblockN{1\textsuperscript{st} Nirupam Basak}
\IEEEauthorblockA{\textit{Indian Statistical Institute}\\
Kolkata 700108, India \\
nirupambasak2020@iitkalumni.org}
\and
\IEEEauthorblockN{2\textsuperscript{nd} Goutam Paul}
\IEEEauthorblockA{\textit{Indian Statistical Institute}\\
Kolkata 700108, India \\
goutam.paul@isical.ac.in}
}

\maketitle

\begin{abstract}
In this work, we develop an efficient decoding method for graph codes, a class of stabilizer quantum error-correcting codes constructed from graph states. While optimal decoding is generally NP-hard, we propose a faster decoder exploiting the structural properties of the underlying graph states. Although distinct error patterns may yield the same syndrome, we demonstrate that the post-measurement state follows a well-defined structure determined by the projective syndrome measurement. Building on this idea, we introduce a hierarchical decoder in which each level can be solved in polynomial time. Additionally, this decoder achieves optimal decoding performance at the lower levels of the hierarchy. This strategy avoids the need for full maximum-likelihood decoding of graph codes. Numerical results illustrate the efficiency and effectiveness of the proposed approach.
\end{abstract}


\section{\label{sec:intro}Introduction}

Despite their potential to outperform classical computers in certain computational tasks, quantum systems are inherently fragile, as unavoidable interactions with the surrounding environment lead to decoherence and operational errors~\cite{zurek2003decoherence, preskill2018quantum}. These errors grow rapidly and pose a fundamental challenge to large-scale, fault-tolerant quantum computation. Consequently, building reliable quantum computers requires correcting these errors through quantum error correction (QEC)~\cite{shor1995scheme, steane1996error, gottesman1997stabilizer, knill1997theory, fowler2012surface, aharonov1997fault, gottesman2002introduction, lidar2013quantum}. QEC protects quantum information by encoding it across multiple physical qubits~\cite{bennett1996mixed, Kitaev1997quantum, Kitaev2003fault, fletcher2007optimum, fletcher2008structured, chao2018quantum, zurek1991decoherence, lloyd1993potentially, arute2019quantum, monroe2014large}.

Graph codes~\cite{schlingemann2001quantum, bell2014experimental_QEC, looi2008quantum, liao2022topological, grassl2002graphs, mackay2004sparsegraph} are a class of quantum error-correcting codes, where quantum information is encoded using the stabilizer formalism~\cite{hwang2016ontherelation, schlingemann2002stabilizer} defined by an underlying graph structure, yielding an $\llbracket n,k,d\rrbracket$ code that encodes $k$ logical qubits into $n$ physical qubits with distance $d$. In this representation, vertices correspond to physical qubits and edges denote entangling interactions, typically controlled-$Z$ operations, used to prepare the underlying graph state. Many important stabilizer codes, including topological and surface codes, can be expressed as graph codes up to local Clifford equivalence~\cite{schlingemann2001quantum, bell2023optimizing, bell2014experimental_QEC, looi2008quantum, liao2022topological}, making graph codes a useful and unifying framework for analyzing and designing quantum error correction schemes.

Stabilizer measurements, also known as syndrome measurements, produce characteristic patterns, called syndromes, that correspond to different error configurations~\cite{divincenzo1996faulttolerant, gottesman1997stabilizer, gottesman1998theory}. Decoding is the process of determining the most likely underlying error for a given syndrome. The performance of a quantum computer is therefore strongly influenced by the decoder~\cite{dennis2002topological, fowler2012surface, terhal2015quantum}.  In general, solving such a maximum-likelihood decoding (MLD)~\cite{poulin2006optimal, ferris2014tensor} problem is NP-hard~\cite{kuo2012hardness, hsieh2011np, berlekamp2003inherent}. This computational intractability highlights the need for decoding strategies that exploit the structure of specific code families.

{\em Our Contribution.}

In this work, we address the above challenge by analyzing the effect of noise on the underlying graph states through syndrome measurement. By leveraging the stabilizer formalism of the underlying graph state, we propose a decoder for graph codes with one logical qubit that can successfully decode any error pattern. More importantly, we find that even when arbitrary errors occur and syndromes are no longer unique for the error patterns, decoding and error correction are still possible using a very simple rule. The original logical state can be recovered by applying only Pauli $Z$ operations at the vertices where the measured syndrome equals $-1$. This leads to a significant simplification of the decoding process.

The results presented in this paper reduce the computational complexity of error correction with graph codes and offer practical advantages for large-scale quantum networks and distributed quantum information processing.

{\em Paper Outline. }

In Section~\ref{sec:prelim}, we briefly review graph states and graph codes, along with the standard MLD problem. Section~\ref{sec:single_noise} then presents the effects of single-qubit noise on graph states, together with the corresponding error syndromes. The impact of arbitrary noise is analyzed in Section~\ref{sec:rand_noise}, where we also show that such noise can be eliminated using a simple phase-correction strategy. In Section~\ref{sec:decode_graph_code}, we present a simple optimal decoder for a graph code correcting any arbitrary error. Numerical results are presented in Section~\ref{sec:numerics}. Finally, Section~\ref{sec:concl_future} concludes the paper and outlines directions for future work.

\section{\label{sec:prelim}Preliminaries}

A \emph{graph state}~\cite{nest2004graphical} is a pure multipartite entangled quantum state associated with an undirected simple finite graph $G=(V,E)$, where $V\subset\mathbb{N}$ is the finite set of vertices and $E\subseteq\{\{i, j\}:i,j\in V,i\neq j\}$ is the set of edges. It is defined as
\begin{equation}
\label{eq:graph_state}
\ket{G}:=\left(\prod_{\{i,j\}\in E}CZ_{i,j}\right)\ket{+}^{\otimes V},
\end{equation}
where $CZ_{i,j}$ is the controlled-$Z$ operation between the vertices $i$ and $j$ and $\ket{+}:=\frac{\ket{0}+\ket{1}}{\sqrt{2}}$.

A \emph{stabilizer} of the graph state~\eqref{eq:graph_state} is an operator $S$ such that $S\ket{G}=\ket{G}$. The \emph{stabilizer group} for the graph state is generated by $K^G_g=\{S_1, S_2, \dots, S_n\}$~\cite{hein2006entanglement, zhang2009stabilizer, kruszynska2006entanglement, lobl2025transforming}, with
\begin{equation}
\label{eq:stab_generator}
S_i=X_iZ_{N_i},\;\forall i\in V,
\end{equation}
where $N_i:=\{j\in V:\{i,j\}\in E\}$ is the neighborhood of $i$ and
\begin{equation*}
\begin{aligned}
&Z_{N_i}:=\prod_{j\in N_i}Z_j=\prod_{j=1}^nZ_j^{I_E(\{i,j\})},\\
&I_{E}(\chi):=\begin{cases}
1&\text{if }\chi\in E,\\
0&\text{otherwise.}
\end{cases}
\end{aligned}
\end{equation*}
Note that for all $i\in V\text{ and }A\subseteq V$
\allowdisplaybreaks
\begin{align}
&S_iZ_A=X_iZ_{N_i}Z_A=X_iZ_AZ_{N_i}\notag\\
&=(-1)^{I_A(i)}Z_AX_iZ_{N_i}=(-1)^{I_A(i)}Z_{N_i}S_i,\label{eq:S_Z}\\
\text{and }&S_iX_A=X_iZ_{N_i}X_A=(-1)^{|A\cap N_i|}X_iX_AZ_{N_i}\notag\\
&=(-1)^{|A\cap N_i|}X_AX_iZ_{N_i}=(-1)^{|A\cap N_i|}X_AS_i.\label{eq:S_X}
\end{align}

Let $K:=\langle S_1, S_2, \dots, S_n\rangle,\text{ for }n=|V|$, be the stabilizer group of the graph state $\ket{G}$. The syndrome measurement of operator $S_i$ is a projective measurement that projects the quantum state onto one of its eigenspaces and yields a classical syndrome bit corresponding to the eigenvalue $(\pm1)$. The projectors are given by~\cite{hein2004multiparty, dangniam2020optimal, zhang2009stabilizer, zhou2008irreducible, dutta2024smallest}
\begin{equation}
\label{eq:stab_projectors}
P_{s_i}:=\frac{1}{2}(I+s_iS_i),
\end{equation}
producing a syndrome bit $s_i=\pm1$. Upon measurement of all the stabilizer generators, a state $\ket{\Psi}$ would be projected onto
\begin{equation}
\ket{\Psi}':=\frac{P_s\ket{\Psi}}{\sqrt{\tr[P_s\ketbra{\Psi}]}},
\end{equation}
where
\begin{equation}
\label{eq:projector}
P_s:=\prod_{i=1}^nP_{s_i}=\prod_{i=1}^n\frac{1}{2}(I+s_iS_i)
\end{equation}
is the projector, producing the syndrome $s=(s_1, s_2, \cdots, s_n)\in\{\pm1\}^n$. Then for any $i$,
\begin{align*}
P_{s_i}\ket{G}&=\frac{1}{2}(I+s_iS_i)\ket{G}
=\frac{1}{2}(1+s_i)\ket{G},
\end{align*}
as $S_i\ket{G}=\ket{G}$. Note that for a valid state, $s_i$ must be $+1$ for all $i$. Thus, $P_{s_i}\ket{G}=\ket{G}$ for all $i$, and hence
\begin{equation}
P_s\ket{G}=
\left(\prod_{i=1}^nP_{s_i}\right)\ket{G}=\ket{G},\label{eq:pure_projection}
\end{equation}
producing the syndrome $s=(+1, +1, \cdots, +1)$. Let us denote this syndrome by $s^{(0)}$.

A graph code~\cite{schlingemann2001quantum, bell2014experimental_QEC, looi2008quantum, liao2022topological, grassl2002graphs, mackay2004sparsegraph}, $\mathcal{C}$, is an $\llbracket n, k, d\rrbracket$ stabilizer code defined using a graph state $\ket{G}$ as follows. Let $S_i$ be any stabilizer generator of $\ket{G}$. Define a logical $Z$ operator as $\widebar{Z}:=S_i$. A logical $X$ operator $\widebar{X}$ may be any operator that anti-commutes with $\widebar{Z}$. Two obvious choices for $\widebar{X}$ are $Z_i$ and $Z_1Z_2\cdots Z_n$, where the former one gives a distance 1 code, and the latter one gives a code with distance $\geq3$ for $n\geq 5$. Thus, the logical qubits are given by
\begin{equation}
\label{eq:logical_state}
\ket{0}_L:=\ket{G},\;\ket{1}_L:=\widebar{X}\ket{G}.
\end{equation}

Suppose $K^G_g=\{S_1, S_2,\dots, S_n\}$ is the set of stabilizer generators of $\ket{G}$. Let us define two sets $K^{(1)}$ and $K^{(2)}$ as
\begin{equation}
\begin{aligned}
K^{(1)}&=\{S_i:[S_i,\widebar{X}]=0,S_i\in K^G_g\},\\
K^{(2)}&=\{S_i:\{S_i,\widebar{X}\}=0,S_i\in K^G_g\},
\end{aligned}
\end{equation}
where $[,]$ and $\{,\}$ are the commutator and anti-commutator operators, respectively. Therefore, for any $S_i\in K^{(1)}$, $S_i\ket{0}_L=\ket{0}_L$ and $S_i\ket{1}_L=\ket{1}_L$, ensuring $S_i\in K^{(1)}$ as a stabilizer of the codespace. However, any $S_i\in K^{(2)}$ produces a $-1$ eigenvalue when it acts on $\ket{1}_L$. Since for any three matrices with $\{A,B\}=0=\{A,C\}$, $ABC=-BAC=BCA$, for $S_i,S_j\in K^{(2)}$, $S_iS_j$ commutes with $\widebar{X}$ and thus becomes a stabilizer of the codespace. As $\widebar{Z}=S_i\in K^{(2)}$, a complete set of stabilizer generators for the codespace is given by
\begin{equation}
K^\mathcal{C}_g=K^{(1)}\cup\{\widebar{Z}S_j:S_j\in K^{(2)}\}.
\end{equation}
Note that $\widebar{Z}$ commutes with every element in $K^\mathcal{C}_g$. This gives a graph code with $1$ logical qubit. Repeating the same process $k$-times, one can get a graph code with $k$ logical qubits~\cite{hwang2016ontherelation}.

To decode an $\llbracket n,k,d\rrbracket$ quantum error-correcting code, the standard strategy~\cite{demarti2024decoding, basak2026hierarchical} is to identify the most likely error pattern $e \in \{0,1\}^n$ from the measured syndrome $s' \in \{0,1\}^m$, where $m = n-k$. The syndromes are determined by the parity-check matrix $H \in \{0,1\}^{m \times n}$ through the relation $s' = He\pmod{2}$. Note that the syndrome $s'$ defined here and the syndrome $s$ defined by~\eqref{eq:stab_projectors} are equivalent as $s'_i=\frac{1+s_i}{2},\;\forall i$.

Assuming that an error on the $i$-th qubit occurs independently with probability $p_i$, the probability of an error pattern $e = (e_1, e_2, \dots, e_n)$ is given by
\[
\Pr[e] = \prod_{i=1}^n (1-p_i)^{1-e_i} p_i^{e_i}
= \Big(\prod_{i=1}^n (1-p_i)\Big)\prod_{i=1}^n \Big(\frac{p_i}{1-p_i}\Big)^{e_i}.
\]
The most likely error pattern $e^*$ is then obtained by solving the following MLD problem:
\begin{equation}
\label{eq:MLD_problem}
\begin{gathered}
e^* = \arg\min_e \sum_{i=1}^n e_i \gamma_i \\
\mathrm{s.t.}\quad
He = s \quad \text{and} \quad
e_i \in \{0,1\},\ i \in \{1,\dots,n\},
\end{gathered}
\end{equation}
where $\gamma_i = \log \frac{1-p_i}{p_i}$.

For simplicity, we assume $p_i = p$ for all $i$, which yields $\gamma_i = \log \frac{1-p}{p}$ and reduces the objective to
\begin{equation}
\label{eq:MLD_problem_simplified}
\arg\min_e \sum_{i=1}^n e_i = \arg\min_e w(e),
\end{equation}
where $w(e)$ denotes the Hamming weight of $e$, i.e., the number of ones in $e$. Note that MLD is NP-hard in general~\cite{hsieh2011np, berlekamp2003inherent}.

\section{\label{sec:single_noise}Graph States under Single-qubit Noise}

In this section, we first find the syndromes corresponding to single-qubit bit-flip and phase-flip errors for a general graph state by measuring the stabilizers. Then we show that this measurement projects any arbitrary noise onto a combination of the bit-flip and phase-flip noise. Thus, to correct any arbitrary noise, it is sufficient to correct only these two noises.

\textbf{Bit-flip noise: }A bit-flip noise on a vertex $v$ flips the state of that vertex in the computational basis ($\ket{0}\leftrightarrow\ket{1}$) with some probability $q$. Under this noise, the state~\eqref{eq:graph_state} becomes an ensemble of pure states $\{(\sqrt{1-q}, \ket{G}), (\sqrt{q}, X_v\ket{G})\}$. As shown in~\eqref{eq:pure_projection}, the state $\ket{G}$ would not change under syndrome measurement and produce the syndrome $s_i=1,\;i=1, 2, \dots, n$. On the other hand, for any stabilizer $S_i$, $S_iX_v=(-1)^{I_{N_i}(v)}X_vS_i$. Then
\allowdisplaybreaks
\begin{align*}
P_{s_i}X_v\ket{G}&=\frac{1}{2}(I+s_iS_i)X_v\ket{G}\\
&=\frac{1}{2}\left(1+(-1)^{I_{N_i}(v)}s_i\right)X_v\ket{G},
\end{align*}
producing the syndrome bit as $-1$ if $i\in N_v$ and $+1$, otherwise. Thus, the syndrome measurement on all the vertices would produce a syndrome $s=(s_1, s_2, \cdots, s_n)$, where
\begin{equation}
\label{eq:synd_bit-flip}
s_i=\begin{cases}
-1,&\text{if }i\in N_v,\\
+1,&\text{otherwise,}
\end{cases}
\end{equation}
and the projected state would be $P_sX_v\ket{G}=X_v\ket{G}.$ Let us denote this syndrome by $s^{(b_v)}$. Therefore, the syndrome measurement does not change the ensemble. However, it produces syndromes $s^{(0)}$ and $s^{(b_v)}$ with probabilities $1-q$ and $q$, respectively. Now, by applying the Pauli operation $I$ or $X_v$, for syndromes $s^{(0)}$ and $s^{(b_v)}$, respectively, one can get back the original state $\ket{G}$ with probability $1$.

\textbf{Phase-flip noise: }A phase-flip noise on a vertex $v$ flips the phase of that vertex in the computational basis ($\ket{+}\leftrightarrow\ket{-}$) with some probability $q$. Under this noise, the state~\eqref{eq:graph_state} becomes an ensemble of pure states $\{(\sqrt{1-q}, \ket{G}), (\sqrt{q}, Z_v\ket{G})\}$. As for any stabilizer $S_i$, $S_iZ_v=(-1)^{\delta_{iv}}Z_vS_i$, we have
\begin{align*}
P_sZ_v\ket{G}&=\prod_{i=1}^n\frac{1}{2}(I+s_iS_i)Z_v\ket{G}\\
&=\prod_{i=1}^n\frac{1}{2}(1+(-1)^{\delta_{iv}}s_i)Z_v\ket{G}.
\end{align*}
Thus, for the state $Z_v\ket{G}$, the syndrome becomes
\begin{equation}
\label{eq:synd_phase-flip}
s_i=(-1)^{\delta_{iv}}
\end{equation}
and the projected state would be $P_sZ_v\ket{G}=Z_v\ket{G}$. Let us denote this syndrome by $s^{(p_v)}$. Therefore, the syndrome measurement does not change the ensemble. Now, similar to the bit-flip error correction, by applying the Pauli operation $I$ or $Z_v$, for syndromes $s^{(0)}$ and $s^{(p_v)}$, respectively, one can get back the original state $\ket{G}$ with probability $1$.

\textbf{Arbitrary noise: }Any arbitrary single-qubit noise can be modeled by the Kraus operators~\cite{tong2006kraus} $E^{(0)}, E^{(1)}, \dots, E^{(m)}$ for some positive integer $m$, satisfying the completeness property 
\begin{equation}
\label{eq:completeness_prop}
\sum\limits_{j=0}^m\left(E^{(j)}\right)^\dagger E^{(j)}=I.
\end{equation}
The action of this noise on vertex $v$ would change the density operator $\rho_G=\ketbra{G}$, corresponding to the state~\eqref{eq:graph_state}, to
\begin{equation}
\label{eq:kraus_act}
\rho'_G=\sum_{j=0}^mE^{(j)}_v\rho_G\left(E^{(j)}_v\right)^\dagger.
\end{equation}
Now, as Pauli operators $X, Z$ and $ZX$, along with $I$ form a basis, we have
\begin{equation}
E^{(j)}_v=a_jI+b_jX_v+c_jZ_v+d_jZ_vX_v,\;\forall j
\end{equation}
for some $a_j, b_j, c_j, d_j\in\mathbb{C}\;\forall j$. The completeness property~\eqref{eq:completeness_prop} gives $\sum\limits_j(a_ja_j^\dagger+b_jb_j^\dagger+c_jc_j^\dagger+d_jd_j^\dagger)=1$. Thus, we have
\begin{align}
&P_sE^{(j)}_v\ket{G}=P_s(a_jI+b_jX_v+c_jZ_v+d_jZ_vX_v)\ket{G}\notag\\
=&a_jP_s\ket{G}+b_jP_sX_v\ket{G}+c_jP_sZ_v\ket{G}+d_jP_sZ_vX_v\ket{G}.
\end{align}

Thus, upon measuring the state~\eqref{eq:kraus_act}, we can get the syndrome and projected state as follows. If the measured syndrome is found to be $s^{(0)}, s^{(b_v)}$ or $s^{(p_v)}$, then the state projects onto $\ket{G}, X_v\ket{G}$ or $Z_v\ket{G}$, respectively. On the other hand, the projection $P_sZ_vX_v\ket{G}$ would produce syndrome $s^{(p_v)}\odot s^{(b_v)}$, where $\odot$ denotes bit-wise multiplication, projecting the state onto $Z_vX_v\ket{G}$.

Thus, we can say that the syndrome measurement projects arbitrary single-qubit noise into a single-qubit Pauli noise. After getting the syndrome, one can apply an appropriate Pauli operation to get back the original state $\ket{G}$ with probability $\sum\limits_j(a_ja_j^\dagger+b_jb_j^\dagger+c_jc_j^\dagger+d_jd_j^\dagger)=1$.

\section{\label{sec:rand_noise}Correcting Arbitrary Noise on Graph State by Phase Correction}

So far, we have assumed that noise affects only a single vertex. Under this simplified assumption, there is a one-to-one correspondence between the generated syndrome and the resulting projected state. This correspondence, however, does not generally hold. In practice, noise may act on any vertex in the system.

Let $v$ be any vertex. Consider phase-flip noise on each vertex in $N_v$. Then the projected state after syndrome measurement would be
\begin{align*}
&P_sZ_{N_v}\ket{G}=\prod_{j=1}^n\frac{1}{2}(I+s_jS_j)Z_{N_v}\ket{G}\\
=&Z_{N_v}\prod_{j=1}^n\frac{1}{2}\left(I+(-1)^{I_{N_v}(j)}s_jS_j\right)\ket{G}\\
=&Z_{N_v}\prod_{j=1}^n\frac{1}{2}\left(1+(-1)^{I_{N_v}(j)}s_j\right)\ket{G}=Z_{N_v}\ket{G},
\end{align*}
producing syndrome as $s_j=-1$, if $j\in N_v$ and $1$, otherwise. Here, the second equality holds due to~\eqref{eq:S_Z}. Note that this syndrome is the same as $s^{(b_v)}$ as defined by~\eqref{eq:synd_bit-flip}, denoting a bit-flip error at vertex $v$. Although the syndrome $s^{(b_v)}$ denotes two different errors on $\ket{G}$, this can be uniquely corrected as follows.

\emph{Case I. }Let the error be phase-flips on each vertex in $N_v$. Then the projected state being $Z_{N_v}\ket{G}$, the operation $U=Z_{N_v}$ would correct it as $Z_{N_v}Z_{N_v}\ket{G}=\ket{G}$.

\emph{Case II. }Let the error be a bit-flip on vertex $v$. Then the projected state would be $X_v\ket{G}$. Therefore, $UX_v\ket{G}=Z_{N_v}X_v\ket{G}=X_vZ_{N_v}\ket{G}=S_v\ket{G}=\ket{G}$.

Thus, we can consider $U=Z_{N_v}$ to be the correction operation when the syndrome is given by $s^{(b_v)}$. This motivates us for the following theorem.

\begin{theorem}
\label{th:graph_state_Pauli}
Let $s = (s_1, s_2, \cdots, s_n)$ denote the syndrome obtained from syndrome measurements performed on a graph state affected by arbitrary Pauli errors. The state can be corrected by applying a $Z$ operator to each vertex $v$ for which $s_v = -1$.
\end{theorem}

Since it directly follows from the proof of Theorem~\ref{th:graph_code_Pauli} in Section~\ref{sec:decode_graph_code}, we omit the proof here. 

As shown in Section~\ref{sec:single_noise}, syndrome measurements project arbitrary noise onto Pauli noise. Thus, the above theorem also holds for any arbitrary noise, and we can have the following corollary.

\begin{corollary}
\label{eq:cor_graph_state}
Let $s = (s_1, s_2, \cdots, s_n)$ denote the syndrome obtained from syndrome measurements performed on a graph state affected by any arbitrary noise. The state can be corrected by applying a $Z$ operator to each vertex $v$ for which $s_v = -1$.
\end{corollary}

\section{\label{sec:decode_graph_code}Decoder for Graph Codes}

As the stabilizers of a graph code depend on the stabilizers of the underlying graph state, we can use the idea of error-correction for the graph state, described in Section~\ref{sec:rand_noise}, to design a decoder for the graph code. Suppose $\mathcal{C}$ is an $\llbracket n, 1, d\rrbracket$ graph code defined using a graph state $\ket{G}$. Then any state $\ket{\Psi}_L\in\mathcal{C}$ can be written as
\begin{equation}
\ket{\Psi}_L=\alpha\ket{0}_L+\beta\ket{1}_L,\;|\alpha|^2+|\beta|^2=1,
\end{equation}
where $\ket{0}_L$ and $\ket{1}_L$ are defined as~\eqref{eq:logical_state}. Also assume that
\begin{equation}
\label{eq:stabilizer_code}
K^\mathcal{C}_g=\{K_1, K_2, \cdots, K_{n-1}\},
\end{equation}
is the set of stabilizer generators of $\mathcal{C}$, where $K_i:=S_n^{I_Q(i)}S_i$, where $Q=\{i:q<i<n\}$ for some $0\leq q\leq n-2$, with $K^G_g=\{S_1, S_2, \dots, S_n\}$ being the set of stabilizer generators of $\ket{G}$. Then the logical $Z$ and $X$ operators would be $\widebar{Z}=S_n,\text{ and }\widebar{X}=Z_{Q\cup\{n\}}$. Now, we prove the following theorem.

\begin{theorem}
\label{th:graph_code_Pauli}
Suppose $\ket{\Psi}'_L$ denotes the state obtained after a Pauli error acts on $\ket{\Psi}_L \in \mathcal{C}$. Let $s=(s_1, s_2, \cdots, s_{n-1})$ be the syndrome obtained from the syndrome measurement performed on $\ket{\Psi}'_L$. Then the corresponding correction operator would be
\begin{equation}
U_{Corr}=S_\mathcal{I}\widebar{X}^aZ_{V^-},
\end{equation}
where $a\in \{0, 1\}$, $\mathcal{I}\subseteq V$ is some index set and $V^-:=\{v\in V|s_v=-1\}$.
\end{theorem}

\begin{proof}
Without loss of generality, assume
\begin{equation}
\ket{\Psi}'_L=Z_{V_z}X_{V_x}\ket{\Psi}_L,
\end{equation}
for some $V_z, V_x\subseteq V$. Then, for $1\leq v<n$, we have
\begin{align}
&\frac{1}{2}(I+s_vK_v)Z_{V_z}X_{V_x}\ket{\Psi}_L\notag\\
=&\frac{1}{2}\left(I+s_vS_n^{I_Q(v)}S_v\right)Z_{V_z}X_{V_x}\ket{\Psi}_L\notag\\
=&\frac{1}{2}\left(I+(-1)^{t(v)}s_vS_n^{I_Q(v)}\right)Z_{V_z}X_{V_x}\ket{\Psi}_L\notag\\
=&\frac{1}{2}\left(I+(-1)^{t(v)+I_Q(v)\cdot t(n)}s_v\right)Z_{V_z}X_{V_x}\ket{\Psi}_L,
\end{align}
with
\begin{equation}
\label{eq:t}
t(i):=I_{V_z}(i)+|V_x\cap N_i|\pmod{2},\;1\leq i\leq n.
\end{equation}
Here, the first equality is due to the definition of $K_v$, and the next two equalities are for~\eqref{eq:S_Z} and~\eqref{eq:S_X}. Therefore, the syndrome measurement would project the state $\ket{\Psi}'_L$ onto itself with syndrome $s_v=(-1)^{t(v)+I_Q(v)\cdot t(n)}$. Then, defining $T:=\{v\in V:t(v)=1\}$, we have
\begin{align}
V^-&=\{v<n:t(v)+I_Q(v)\cdot t(n)=1\}\notag\\
\text{and }Z_{V^-}&=\prod_{v=1}^{n-1}Z_v^{t(v)+I_Q(v)\cdot t(n)}\notag\\
&=\left(\prod_{v\in V}Z_v^{t(v)}\right)Z_n^{t(n)}\prod_{v\in Q}Z_v^{t(n)}\notag\\
&=Z_T\left(Z_{Q\cup\{n\}}\right)^{t(n)}
\end{align}

Also,
\begin{align}
Z_TZ_{V_z}X_{V_x}&=Z_{V_z}\left(\prod_{v\in V}Z_v^{|V_x\cap N_v|}\right)Z_{V_z}X_{V_x}\notag\\
&=\left(\prod_{v\in V_x}Z_{N_v}\right)\prod_{v\in V_x}X_v\notag\\
&=S_{V_x},\\
\text{and }\;Z_{V^-}Z_{V_z}X_{V_x}&=\left(Z_{Q\cup \{n\}}\right)^{t(n)}S_{V_x}=X^{t(n)}_LS_{V_x},
\end{align}
indicating that
\begin{equation}
\label{eq:correction_op_graph_code}
U_{Corr}=S_{V_x}\widebar{X}^{t(n)}Z_{V^-}
\end{equation}
is the correction operation for the error $Z_{V_z}X_{V_x}$.\\
This completes the proof.
\end{proof}

Since the syndrome measurement projects arbitrary noise onto a Pauli noise, this theorem also holds for arbitrary noise.

Note that $V_x$ and $t(n)$, which depend on the noise, are two unknown quantities. Therefore, one cannot directly apply the correction operation~\eqref{eq:correction_op_graph_code} to correct an error. However, this provides a structural formulation of correction operations that can be used to build a hierarchical decoder for a graph code, as shown in Algorithm~\ref{alg:decoder}.

\begin{algorithm}
\caption{\label{alg:decoder}Graph Code Decoder}
\begin{algorithmic}[1]
\Require $\widebar{X},\;K^G_g=\{S_1, S_2, \dots, S_n\}\text{ and }s=(s_1, s_2, \cdots, s_{n-1}).$
\State Calculate $Z_{V^-}:=\prod\limits_{i\in V^-}Z_i$ and $\widebar{X}Z_{V^-},$ where $V^-:=\{v|s_v=-1\}.$
\State $wt\leftarrow n$ and $\ell\leftarrow0.$
\While{$\ell<wt$}
\State $U_{Corr}\leftarrow\arg\min\{w(S_\mathcal{I}Z_{V^-}),w(S_\mathcal{I}\widebar{X}Z_{V^-}):\mathcal{I}\subseteq V,\,|\mathcal{I}|=\ell\}\}.$
\State $wt\leftarrow\min\{w(S_\mathcal{I}Z_{V^-}), w(S_\mathcal{I}\widebar{X}Z_{V^-}):\mathcal{I}\subseteq V,\,|\mathcal{I}|=\ell\}.$
\State $\ell\leftarrow\ell+1.$
\EndWhile
\State\Return $U_{Corr}.$
\end{algorithmic}
\end{algorithm}

Here, Algorithm~\ref{alg:decoder} would find the correction operation with the lowest weight of Pauli operators, similar to the MLD problem~\eqref{eq:MLD_problem_simplified}. As each $S_i$ contains a $X$ at $i$-th position and $\widebar{X}$ and $Z_{V^-}$ do not contain any $X$, at level $\ell$, $w(S_\mathcal{I}Z_{V^-}), w(S_\mathcal{I}\widebar{X}Z_{V^-})\geq\ell$. Therefore, if $wt\leq\ell$, for some level $\ell$, one cannot get a lower weight at any level $>\ell$. This ensures that the corresponding weight is the minimum weight. As, for level $\ell$, there are $\binom{n}{l}$ possible choices for $\mathcal{I}$, the corresponding run-time would be $n^{\mathcal{O}(\ell)}$. Also, for a distance $d$ code, the maximum weight of any correctable error being $\frac{d-1}{2}$, the maximum level for this algorithm would be $\frac{d-1}{2}$.

%  From binomial theory
% Worst case complexity << $2^{n-2}$
% Average case complexity << $2^{n-3}$

\section{\label{sec:numerics}Numerical Results}

We have considered $G=C_n$, the graphs having a single cycle with $n$ vertices for $n=5, 9$ and $11$. With $\widebar{X}=Z_V$, these would give distance-3 codes for $n=5, 9$ and the smallest distance-5 code for $n=11$~\cite{gheorghiu2010location, yu2007graphical, calderbank1998quantum}. For $n=11$, taking $\widebar{X}=Z_3Z_4\cdots Z_n$ we get a distance-3 code. We analyze the logical error rate $p_\text{L}$ as a function of the physical error rate $p$ and the hierarchy level $\ell$. The performance is compared against MLD formulated as a mixed-integer program (MIP), which achieves exact optimal decoding at the expense of computational cost~\cite{Vanderbei2020integer, feldman2003decoding, feldman2005using}.

In Figure~\ref{fig:plots}, we plot $p_L$ against $p$ for different $\ell$. The logical error for the MIP decoder is also plotted for benchmarking. We found that even $\ell=0$, which runs in $\mathcal{O}(1)$-time, yields some improvement in the logical error rate over no decoding, and the optimal error rate is achieved at low levels, $\ell\sim1\textendash3$.

\begin{figure}[htpb]
\centering
\subfigimg[width=0.5\columnwidth]{a}{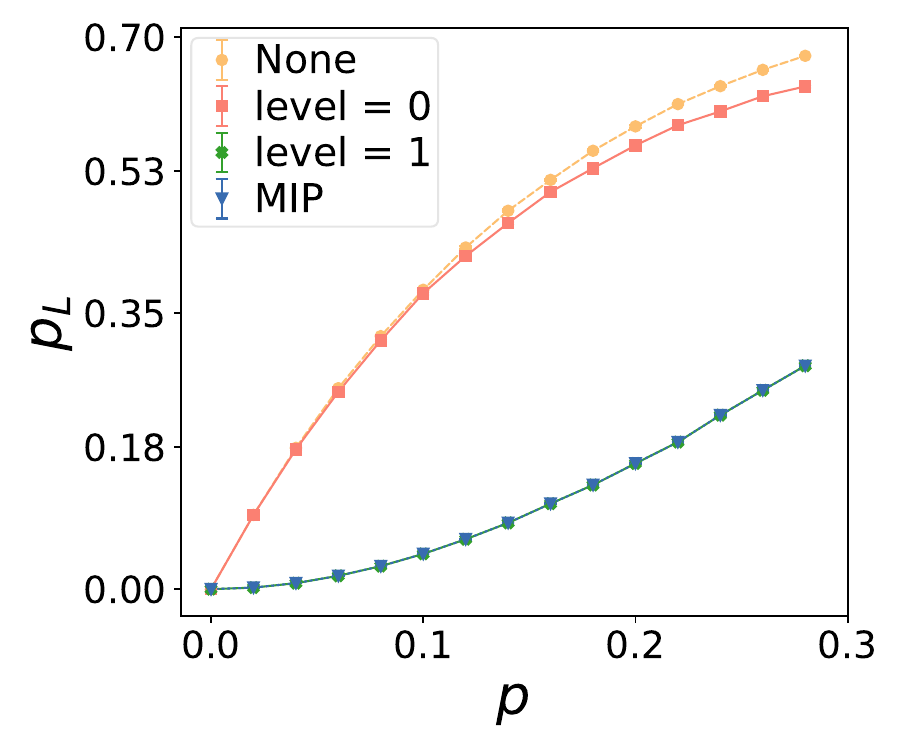}\hfill
\subfigimg[width=0.5\columnwidth]{b}{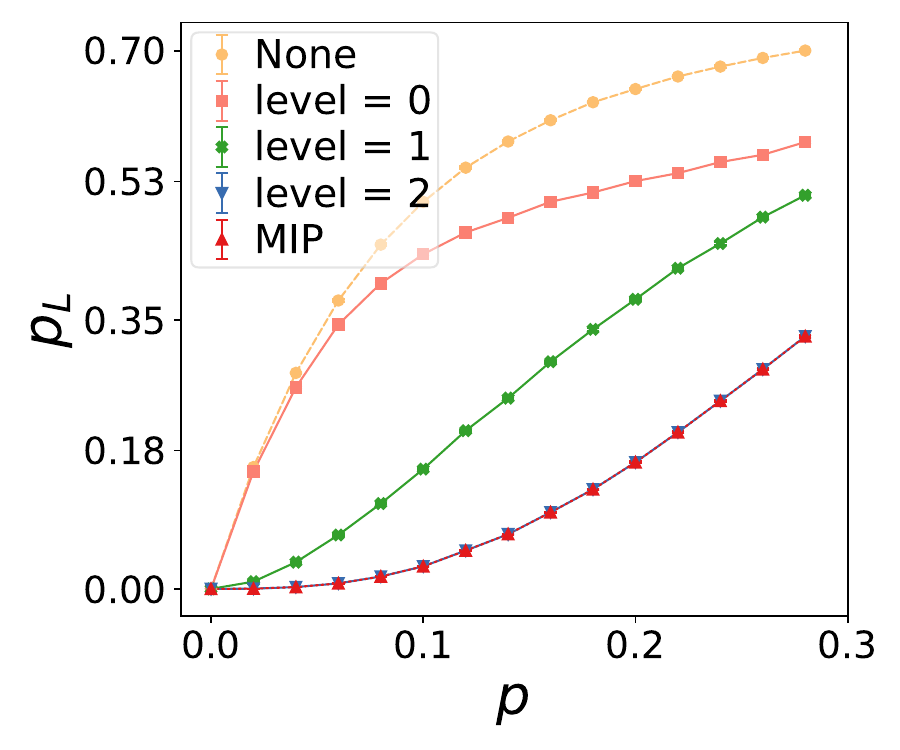}\\
\subfigimg[width=0.5\columnwidth]{c}{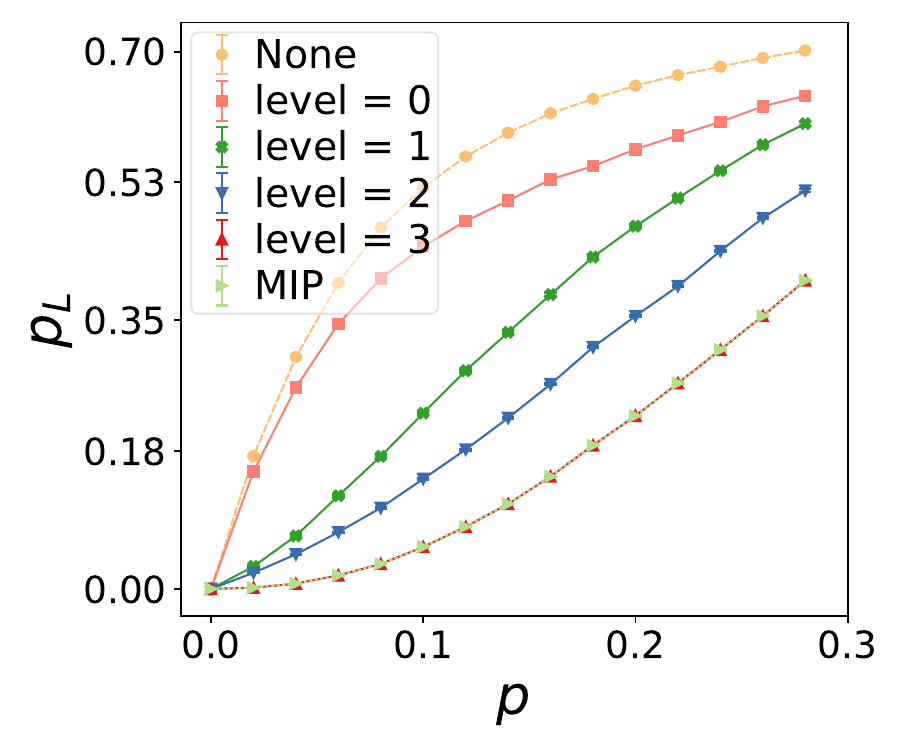}\hfill
\subfigimg[width=0.5\columnwidth]{d}{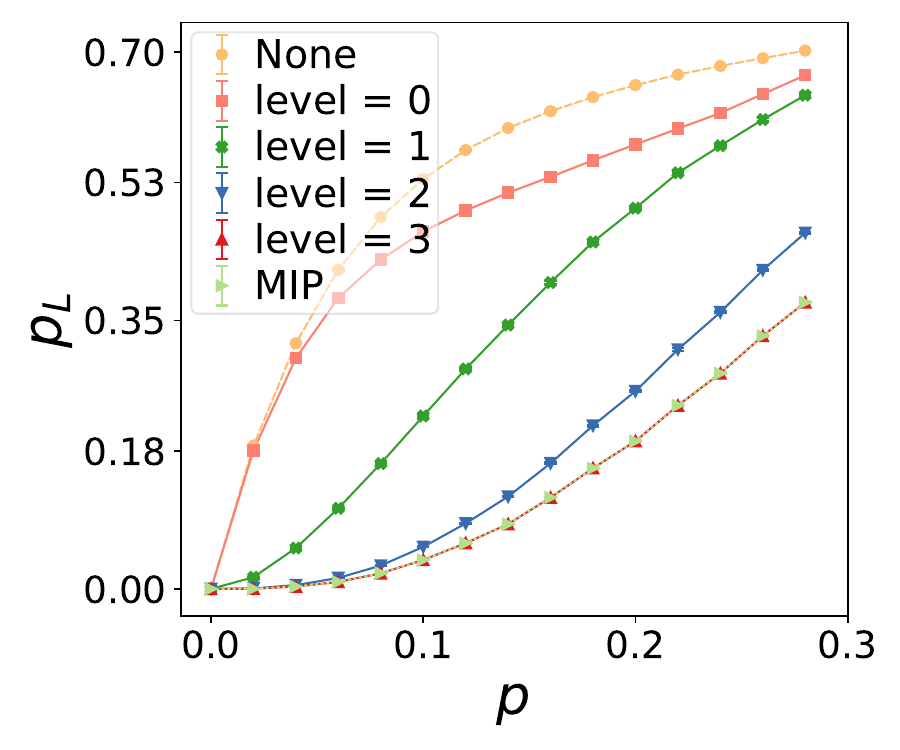}
\caption{Plot of logical error rate $p_L$ for \idg{a} $\llbracket5,1,3\rrbracket$,  \idg{b} $\llbracket9,1,3\rrbracket$,  \idg{c} $\llbracket11,1,3\rrbracket$ and \idg{d} $\llbracket11,1,5\rrbracket$ against physical error rate $p$. Error rates for different levels of the proposed decoder are shown. The logical error rates for the highest level coincide with those of MIP.}
\label{fig:plots}
\end{figure}

In Table~\ref{tab:time}, we present the average time taken by the proposed decoder for optimal decoding. The average time for the MIP decoder is also provided, and the proposed hierarchical decoder is found to require significantly less computational time.

\begin{table}[htpb]
\centering
\caption{Average run-time for optimal decoding using the proposed and the MIP decoder in milliseconds. The proposed hierarchical decoder takes significantly less time compared to the MIP decoder.}
\begin{tabular}{|c|c|c|}
\hline
\multirow{2}{*}{\textbf{QEC code}}&\multicolumn{2}{c|}{\textbf{Optimal Decoding Time}}\\
\cline{2-3}
&\textbf{Proposed Decoder}&\textbf{MIP Decoder}\\
\hline
$\llbracket5,1,3\rrbracket$&0.015&2.593\\
\hline
$\llbracket9,1,3\rrbracket$&0.062&5.643\\
\hline
$\llbracket11,1,3\rrbracket$&0.120&7.363\\
\hline
$\llbracket11,1,5\rrbracket$&0.114&7.563\\
\hline
\end{tabular}
\label{tab:time}
\end{table}

\section{\label{sec:concl_future}Conclusion and Future Work}

\subsection{\label{ssec:concl}Conclusion}

In this article, we study the impact of noise on graph states. We demonstrate that syndrome measurements map any noise onto effective Pauli noise, rendering it correctable via Pauli operations. While a one-to-one correspondence between the error syndrome and the underlying error exists when at most a single error occurs, this correspondence generally breaks down in realistic settings where arbitrary error patterns may arise. In practice, the same syndrome can therefore be associated with multiple distinct error configurations. By exploiting the structure of the underlying graph states, we show that any arbitrary Pauli error pattern can be decoded and optimally corrected using only Pauli $Z$ operations for graph codes with a single logical qubit. The required correction can be inferred directly from the syndrome bits, with a value of $-1$ signaling the presence of an error. This approach significantly reduces the decoding complexity for graph codes.

\subsection{\label{ssec:future}Future Work}

Future research could explore the performance of codes developed from other graph structures. Another direction could be to develop similar decoders for graph codes with multiple logical qubits and to extend this framework to construct decoders for other stabilizer quantum error-correcting codes. 
% Hypergraph LDPC codes

Additionally, more efficient methods for creating and correcting graph states required in quantum networks, multiparty computation, measurement-based quantum computing, and other quantum information processing and communication protocols could be investigated.

\bibliographystyle{IEEEtran}
\bibliography{IEEEabrv,references}

\end{document}